\documentclass{article}
\usepackage{amsfonts}
\usepackage{amsmath}
\usepackage{color}
\usepackage{hyperref}

\newtheorem{theorem}{Theorem}

\newtheorem{lemma}[theorem]{Lemma}

\newenvironment{proof}[1][Proof]{\noindent\textbf{#1.} }{\ \rule{0.5em}{0.5em}}
\newcommand{\ind}{1\hspace{-2.1mm}{1}} 
\newcommand{\RR}{\mathbb{R}}
\newcommand{\BSB}{\overline{\mathrm{BS}}}
\newcommand{\BS}{\mathrm{BS}}
\newcommand{\Ee}{\mathcal{E}}
\newcommand{\Nn}{\mathcal{N}}
\newcommand{\rha}{r}
\newcommand{\Oo}{\mathcal{O}}
\newcommand{\EE}{\mathbb{E}}
\newcommand{\PP}{\mathbb{P}}
\newcommand{\Ds}{\partial_{\sigma}}
\newcommand{\Dss}{\partial^2_{\sigma\sigma}}
\newcommand{\Dx}{\partial_{x}}
\newcommand{\Dxx}{\partial^{2}_{xk}}
\newcommand{\DK}{\partial_{K}}
\newcommand{\DKK}{\partial^{2}_{KK}}
\newcommand{\Dk}{\partial_{k}}
\newcommand{\Dkk}{\partial^{2}_{kk}}
\newcommand{\Dsk}{\partial^{2}_{\sigma k}}

\begin{document}
\title{The implied volatility of Forward-Start options: ATM short-time level, skew and curvature}
\author{Elisa Al\`{o}s
 \\
Dpt. d'Economia i Empresa\\
and Barcelona GSE\\
Universitat Pompeu Fabra\\
c/Ramon Trias Fargas, 25-27\\
08005 Barcelona, Spain 
\and Antoine Jacquier
\\ Department of Mathematics\\
Imperial College London, London SW7 2AZ, UK\\
and\\
Department of Mathematics\\
Baruch College, CUNY, New-York
\and Jorge A. Le\'{o}n
\\
Control Autom\'{a}tico\\
CINVESTAV-IPN\\
Apartado Postal 14-740\\
07000 Ciudad de M\'{e}xico,  Mexico}
\date{}

\maketitle

\begin{abstract}
Using Malliavin Calculus techniques, 
we derive closed-form expressions for the at-the-money behaviour of the forward implied volatility, 
its skew and its curvature, in general Markovian stochastic volatility models with continuous paths.

\vspace{0.2cm}


\end{abstract}

\section{Introduction}

For any fixed evaluation time~$t$, we consider a Forward-Start call option 
(originally introduced by Rubinstein~\cite{Rubinstein}) 
with forward-start date $s>t$ 
and maturity $T>s$, written on some underlying stock price process~$S$;
the particular feature of this option is that it allows the holder to receive, at the future time~$s$ and at no additional cost,
a standard European Call option expirying at~$T$, with strike set to~$KS_s$, for some $K>0$. 
Classical occurrences of Forward-Start options include for example employee stock options 
and cliquet options~\cite{Rubinstein}.
In the Black-Scholes formula, stationarity of the increments implies, 
by a simple conditional expectation argument, that the Forward-Start option corresponds exactly
to a plain Call option with time to maturity~$T-s$. 
This is not true any longer for general stochastic volatility models, though,
and the resetting at the forward-start date makes the whole analysis a lot more subtle.
Lucic~\cite{Lucic} and Musiela and Rutkowski~\cite[I.7.1.10]{MR} 
applied a change-of-measure argument to relate the price of a Forward-Start option 
to that of a standard Call option, albeit with a randomised starting volatility. 
Pricing formulae based on the knowledge of the characteristic function--hence mainly applicable to 
affine models, in the sense of~\cite{DFS}--were derived by Kruse and N\"ogel~\cite{KN} and by Guo and Hung~\cite{Guo}.

Because of this resetting feature, the implied volatility of Forward-Start options 
is substantially different from the usual vanilla smile.
In a series of papers, Jacquier and Roome~\cite{JR1, JR2, JR3} studied these specificities,
and conducted a thorough analysis in the case of the Heston model.
They in particular singled out the explosive nature of the forward smile as the remaining maturity becomes small. 
Their analysis was based on the knowledge of the characteristic function of the underlying process, 
and the asymptotic behaviour thereof.
Recently, Mazzon and Pascucci~\cite{Mazzon} took over the topic and proved an approximation of the out-of-the-money forward smile in multi-factor local stochastic volatility models, using expansions for parabolic equations.
We concentrate here on the at-the-money (ATM) case,
and characterise the short-time limit of the forward implied volatility, its skew and its curvature, for general Markovian stochastic volatility models with continuous paths. 
Using Malliavin Calculus techniques, we show that--contrary to the classical Vanilla case--the ATM 
short-time level is a direct function of the correlation between the underlying and its instantaneous volatility. 
The at-the-money skew depends of the Malliavin derivative of the volatility process, 
in a similar way as for Vanilla options, 
but the curvature decays at the speed $\Oo(T-s)$. 

In Section~\ref{sec:General}, we introduce Forward-Start options and the main notations used throughout the paper, and prove a decomposition formula for the option price,
which we use in Section~\ref{sec:ATM} to compute the asymptotic behaviour of the at-the-money implied volatility level, skew and curvature. 
We apply these formulae in Section~\ref{sec:Example} to a generalised version of the Stein-Stein stochastic volatility model,
and postpone the proofs of the main results to Section~\ref{sec:Proofs}.

\section{Forward-Start options and the decomposition formula}\label{sec:General}
We consider a generic stochastic volatility model over a finite time horizon~$[0,T]$, 
solution to the stochastic differential equation
\begin{equation}\label{eq:stock}
d X_{t} = \left(\rha - \frac{1}{2}\sigma_t^2\right) dt
 + \sigma_{t}\left( \rho dW_{t}^{*}+\sqrt{1-\rho ^{2}}B_{t}^{*}\right),
\end{equation}
where~$\rha$ is the instantaneous interest rate (assumed constant), 
$W^{*}$ and~$B^{*}$ are independent standard Brownian
motions on a given probability space 
$\left( \Omega ,\mathcal{F},\PP^{\ast}\right) $ and~$\sigma$ is a positive, 
square-integrable process, adapted to the filtration generated by~$W^{*}$.
Here, $X$ denotes the logarithm of a stock price, and we consider the dynamics~\eqref{eq:stock}
directly under a given risk-neutral probability measure~$\PP^*$.
We shall denote by~$\mathcal{F}^{W^{*}}$ and~$\mathcal{F}^{B^{*}}$ the
filtrations generated by~$W^{*}$ and~$B^{*}$, 
and define $\mathcal{F}:=\mathcal{F}^{W^{*}}\vee \mathcal{F}^{B^{*}}$.
We are interested here in computing the price of a Forward-Start option, 
the payoff of which, at maturity~$T$, is equal to 
$\left(e^{X_{T}} - e^{\alpha}e^{X_{s}}\right)_{+}$, where $s\in [0,T]$ is the forward-start date
and~$\alpha\in\mathbb{R}$ the log-forward moneyness.
This is usually called Type-II Forward-Start, and we refer the reader to~\cite{Lucic}
for details about Type-I and the symmetries between both types.
Classical no-arbitrage arguments yield that the price at inception $t\leq s$ is given by
\begin{equation}
V_{t}=e^{-\rha(T-t) }\EE^{*}_{t}\left( e^{X_{T}} - e^{\alpha}e^{X_{s}}\right)_{+},  \label{price formula}
\end{equation}
where $\EE^{*}$ denotes the conditional expectation under~$\PP^*$ given~$\mathcal F_t$. 
If $t\geq s$, this is simply a standard European Call option evaluated during the life of the contract.
We shall denote by $\BS(t,x,K,\sigma ) := e^{x}\Nn(d_{+}) - e^{-\rha(T-t)} K \Nn(d_{-})$ 
the price of a European call option in the Black-Scholes model 
with constant volatility~$\sigma$, current log-stock price~$x$, time to maturity $T-t,$ 
strike~$K$ and interest rate~$\rha$.
Here~$\Nn$ is the Gaussian cumulative distribution function, and
$$
d_{\pm }:=\frac{x-\ln K+\rha(T-t)}{\sigma \sqrt{T-t}}\pm \frac{\sigma }{2}\sqrt{T-t}.
$$
The corresponding Black-Scholes differential operator (in the log variable) is denoted by 
$$
\mathcal{L}_{\BS}(\sigma)
 := \partial_{t} + \frac{\sigma^{2}}{2}\Dxx
 + \left(\rha-\frac{\sigma^{2}}{2}\right)\Dx - \rha,
$$
so that in particular $\mathcal{L}_{\BS}\left( \sigma \right) \BS\left( \cdot,\cdot ;K,\sigma \right) =0$.
The inverse of~$\BS$ with respect to volatility shall be denoted by $\BS^{-1}(\cdot):=\BS^{-1}(t,x, K,\cdot)$,
while $\alpha^*:=\rha(T-s)$ represents the at-the-money forward log-moneyness,
which will be a quantity of particular interest.
We will finally make heavy use of the following two functions:
$$
G(t,x,K,\sigma ):=(\Dxx-\Dx)\BS(t,x,K,\sigma)
\quad\text{and}\quad
H(t,x,K,\sigma):=\Dx G(t,x,K,\sigma).
$$
We introduce the strike-adjusted forward process
\begin{align*}
M_{t} & := \EE^{*}_{t}\left(e^{\alpha}e^{X_{s}}\right)
 = e^{\alpha+\rha s+X_{0}}+\int_{0}^{t}\sigma_{u}
 \ind_{\left[0,s\right] }(u)e^{\alpha + X_{u}}e^{\rha(s-u)}\left( \rho
dW_{u}^{*}+\sqrt{1-\rho ^{2}}dB_{u}^{*}\right) \\
 & = M_{0} + e^{\alpha}\int_{0}^{t\wedge s}\sigma_{u}e^{\rha(s-u)}e^{X_{u}}
 \left( \rho dW_{u}^{*}+\sqrt{1-\rho ^{2}}dB_{u}^{\ast
}\right),
\end{align*}
as well as the realised volatility
$v_{t} := \sqrt{Y_{t} / (T-t)}$, with
$Y_{t} := \int_{t\vee s}^{T}\sigma_{u}^{2}du$.
Notice that, if $t<s$, $v_{t}\sqrt{T-t}=v_{s}\sqrt{T-s}.$
In order to prove our main results, we consider the following hypotheses,
ensuring both a unique strong solution to~\eqref{eq:stock} 
and regularity of the solution in the Malliavin sense:\\

\noindent \textbf{(H1)} $\sigma$ is bounded below and above almost surely by strictly positive constants.\\
\textbf{(H2)} Both $\sigma$ and $\sigma e^{X}$ belong to $\mathbb{L}^{2,2}\cap \mathbb{L}^{1,4}$.\\

The space~$\mathbb{L}^{p,q}$ is the classical space on which processes are~$q$ times
Malliavin differentiable in~$L^p$.
We refer the reader to~\cite[Section 1.2]{Nualart} for full details.
We are now in a position to prove the following decomposition theorem, 
which identifies the impact of correlation on 
the price of Forward-Start options.
To do so, let us introduce the auxiliary process
$\Lambda_{u}^{W^{\ast}} := \int_{u\vee s}^{T}D_{u}^{W^{*}}\sigma
_{\theta }^{2}d{\theta }$, for $u\in[0,T]$,
where~$D_{u}^{W^{*}}$ denotes the Malliavin derivative with respect to the Brownian motion~$W^*$.
One of the main results of this paper is the following decomposition theorem, 
the proof of which is postponed to Section~\ref{sec:the:hw}.
\begin{theorem}\label{the:hw}
Under~\textbf{(H1)} and~\textbf{(H2)}, for all $0\leq t\leq s\leq T$,
\begin{align*}
V_{t} = \EE^{*}_{t}\Big[ e^{X_{t}}\BS\left( s,0,e^{\alpha},v_{s}\right)
& + \frac{\rho }{2}\int_{s}^{T}e^{-\rha(u-t)}
H(u,X_{u},M_{u},v_{u})\sigma_{u}\Lambda_{u}^{W^{*}}du\\
& + \frac{\rho }{2}G(s,0,e^{\alpha},v_{s})\int_{t}^{s}e^{-\rha(u-t)}e^{X_{u}}\sigma_{u}\Lambda_{u}^{W^{*}}du\Big].
\end{align*}
\end{theorem}

Hypotheses~\textbf{(H1)} and~\textbf{(H2)} are stated here mainly for simplicity,
and the theorem, as can be seen through its proof, still holds under suitable integrability conditions.
In the case $t=s$, the Forward-Start option reduces to a standard European Call option,
and we recover precisely the decomposition formula proved in~\cite{Alos2015}.
If the volatility process is constant, equal to some $\sigma>0$, 
then $v_{s}=\sigma$, $\Lambda ^{W} = 0$ almost surely, 
and, for $t\leq s$,
\begin{equation}\label{TheBScase}
V_{t} = e^{X_{t}}\BS\left( s,0,e^{\alpha},\sigma \right),
\end{equation}
which is the classical Forward-Start option price in Black-Scholes~\cite{Rubinstein, Wilmott}.

\section{At-the-money behaviour of the short-maturity forward smile}\label{sec:ATM}
We now delve into the core of our analysis, and use Theorem~\ref{the:hw} 
to deduce the precise short-maturity behaviour of the at-the-money forward implied volatility smile, 
its skew and its curvature.
For $t\in \lbrack 0,s]$, we define the forward implied volatility $I(t,s;\alpha)$
as the unique non-negative solution to the equation
\begin{equation}\label{elprecio}
V_{t} = e^{X_{t}}\BS\left( s,0,e^{\alpha},I(t,s;\alpha )\right) .
\end{equation}
Obviously, in the constant volatility case $\sigma_{u}=\sigma>0$, so that $I(t,s;\alpha )=\sigma$.
In order to streamline the presentation of the results, we introduce the following quantity, 
which will appear several times:
\begin{equation}\label{eq:Ee}
\Ee_{t,s} := \frac{e^{-X_t}}{2}\EE^{*}_{t}\left( \frac{1}{\sigma_{s}}\int_{t}^{s}e^{-
\rha(u-t)}e^{X_{u}}\sigma_{u}\left(D_{u}^{W^{*}}\sigma_{s}^{2}\right)du\right)
\end{equation}

\subsection{At-the-money smile}
The following theorem, with proof relegated to Section~\ref{sec:el13}, 
provides the short-maturity behaviour for the ATM forward smile.
\begin{theorem}\label{el13}
Under~\textbf{(H1)}-\textbf{(H2)}, for all $t<s$, 
\begin{equation*}
\lim_{T\downarrow s}I(t,s;\alpha ^{*})
 = \EE^{*}_{t}(\sigma_{s}) + \rho \Ee_{t,s}.
\end{equation*}
\end{theorem}

Perhaps not surprisingly, in the uncorrelated case, the short-maturity ATM smile is given
uniquely as the expectation of the future instantaneous volatility.
The correlation parameter acts as a correction term around this level.
In the Vanilla case $s=t$, the theorem recovers exactly the behaviour proved by Durrleman~\cite{Durrleman}. 
For practical purposes, the integral and the expectation can be computed in most (continuous Markovian) models
used in financial practice. 
We shall highlight in Section~\ref{sec:Example} how this looks like exactly in the case of the Stein-Stein stochastic volatility model.

\subsection{At-the-money skew}
The at-the-money level of the forward smile is directly observable.
From a trader's point of view, the skew, i.e. the derivative of the implied volatility with respect to
the (log) moneyness, is a key tool indicating the level of the Put-Call asymmetry.
Consider now the additional hypothesis regarding the regularity of the volatility process:\\

\textbf{(H3)} There exist $C>0$, $\delta \geq 0$, such that, for all $t\leq \theta <u<r$,
$$
\EE^{*}_{t}\left[\left( D_{u}^{W^{*}}\sigma_{r}\right)^{2}\right] \leq C(r-u)^{2\delta}
\quad\text{and}\quad
\EE^{*}_{t}\left[\left( D_{u}^{W^{*}}D_{u}^{W^{*}}\sigma_{r}\right) ^{2}\right]
\leq C\left(\frac{r-u}{r-\theta}\right)^{2\delta}.
$$
The small-maturity at-the-money forward skew is provided in the following theorem, 
proved in Section~\ref{sec:the:pdi}.
\begin{theorem}\label{the:pdi}
Under~\textbf{(H1)}-\textbf{(H2)}-\textbf{(H3)}, assume that there exists a $\mathcal{F}_{s}$-measurable
random variable $\overline{\sigma}_{s}$ such that 
\begin{equation}\label{condition}
\lim_{T\downarrow s}E_{s}^{*}\left[\sup_{s\leq u\leq \theta \leq T}
\EE^{*}_{u}\left(D_{u}^{W^{*}}\sigma_{\theta}^{2} - \overline{\sigma}_s^{2}\right)^{4}\right] = 0.
\end{equation}
Then, for all $s<t$, 
$$
\lim_{T\downarrow s}\frac{\partial I}{\partial \alpha }(t,s;\alpha^{*})
 = {\frac{\rho e^{-\rha(s-t)}}{4e^{X_{t}}}\EE^{*}_{t}
}\left( e^{X_s}{\frac{\overline{\sigma}_{s}^{2}}{\sigma_{s}^{2}}}\right).
$$
\end{theorem}

In the Vanilla case $t=s,$ this formula agrees with the at-the-money short-time limit skew proved in~\cite{alv07}.
Before proving the theorem, let us state the following, short yet useful, result:

\subsection{At-the-money curvature}
We now concentrate our attention to the at-the-money curvature $\frac{\partial ^{2}I}{\partial \alpha ^{2}}(t,s;\alpha ^{*})$ of the forward smile.


\begin{theorem}\label{el14}
Under~\textbf{(H1)} and~\textbf{(H2)}, for all $0\leq t\leq s\leq T$,
\begin{align*}
\lim_{T\downarrow s}(T-s)\frac{\partial ^{2}I}{\partial \alpha ^{2}}(t,s;\alpha ^{*}) 
 & = \frac{1}{4}\EE^{*}_{t}\left[
 \int_{t}^{s}\EE^{*}_{u}\left(\frac{D_{u}^{W^{*}}\sigma_{s}^{2}}{\sigma_{s}}\right)^{2}
 \frac{du}{\EE^{*}_{u}(\sigma_{s})^{3}}\right] \\
  & + \frac{1}{\EE^{*}_{t}(\sigma_{s})}
 - \frac{1}{\EE^{*}_{t}(\sigma_{s}) + \rho \Ee_{t,s}} \\
 & -\frac{\rho }{2}e^{-X_{t}}\EE^{*}_{t}\left( \frac{1}{\sigma_{s}^{3}}
\int_{t}^{s}e^{-\rha(u-t)}e^{X_{u}}\sigma_{u}\left( D_{u}^{W^{*}}\sigma
_{s}^{2}\right) du\right).
\end{align*}
\end{theorem}

\section{Example: the extended Stein-Stein model}\label{sec:Example}
The above formulae for the at-the-money forward implied volatility level, skew and curvature
look daunting. 
However, we emphasise here, through an example widely used in practice, 
that they are in fact fully explicit.
Consider a stochastic volatility model of the form $\sigma=f(Y)$, 
for some $f\in\mathcal{C}^{1,2}(\RR)$ where~$Y$ is an Ornstein-Uhlenbeck process of the form 
\begin{equation}
\label{Stein}
dY_t=\kappa (m-Y_t)dt + \lambda dW_t^*,
\qquad Y_0 \in \RR,
\end{equation}
for some positive constants~$\kappa$,~$m$ and~$\lambda$.
When $\sigma(y)\equiv y$, this model is nothing else than the Stein-Stein stochastic volatility model~\cite{Stein}, 
and the function~$f(\cdot)$ allows for greater flexibility.
For any $t\leq s$, we can then write
\begin{align*}
Y_s & = Y_t e^{-\kappa (s-t)} + m\left(1-e^{-\kappa (s-t)}\right)
 + \lambda\int_t^s e^{-\kappa (s-r)} dW_r^*\\
 & =: g(t,s) + \lambda\int_t^s e^{-\kappa (s-r)} dW_r^*,
\end{align*}
from which we deduce $D_t^{W^*}Y_s = \lambda e^{-\kappa (s-t)}$.
Theorem~\ref{el13} then applies with 
$$
\Ee_{t,s} = \lambda e^{-X_t}\EE^{*}_{t}\left[f'(Y_s)\int_{t}^{s}e^{-\rha(u-t)}e^{X_{u}}f(Y_u) e^{-\kappa (s-u)}du\right].
$$
On the other hand, Theorems~\ref{the:pdi} and~\ref{el14} yield
\begin{align*}
\lim_{T\downarrow s}\frac{\partial I}{\partial \alpha }(t,s;\alpha^{*})
 &  = \frac{\rho \lambda e^{-\rha(s-t)}}{2e^{X_{t}}} \EE^{*}_{t}\left(\frac{ e^{X_s}f'(Y_s)}{f(Y_s)}\right)\\
\lim_{T\downarrow s}(T-s)\frac{\partial ^{2}I}{\partial \alpha ^{2}}(t,s;\alpha ^{*}) 
 & = \lambda^2\EE^{*}_{t}\left(\int_{t}^{s}e^{-2\kappa(s-u)}
\frac{\EE^{*}_{u}\left( f'(Y_s)\right)^{2}}{\EE^{*}_{u}\left(
f(Y_s)\right)^{3}}du\right) \\
 & +\frac{1}{\EE^{*}_{t}\left( f(Y_s)\right) }
 - \frac{1}{\EE^{*}_{t}\left( f(Y_s)\right) + \rho \Ee_{t,s}}
 - \rho \EE^{*}_{t}\left( \frac{\Ee_{t,s}}{f(Y_s)^{2}}\right).
\end{align*}
In the standard Stein-Stein case, where $f(y)\equiv y$ (this function does not exactly satisfy the regularity assumptions,
but we assume it for practical purposes in this example), we can make these computations more explicit, as
$f'(y) = 1$, 
\begin{align*}
\Ee_{t,s} & = \lambda e^{-X_t}\EE^{*}_{t}\left[\int_{t}^{s}e^{-\rha(u-t)}e^{X_{u}} Y_u e^{-\kappa (s-u)}du\right],\\
\EE_t^*(f(Y_s)) & = \int_{\RR} f\left(z\lambda\sqrt{\int_s^t e^{-2\kappa (s-r)}dr}+g(t,s)\right)\phi (z)dz = g(t,s),
\end{align*}
where~$\phi$ is the Gaussian density, so that
\begin{align*}
\lim_{T\downarrow s}I(t,s;\alpha ^{*})
& = g(t,s) + \rho\Ee_{t,s},\\
\lim_{T\downarrow s}\frac{\partial I}{\partial \alpha }(t,s;\alpha^{*})
 &  = \frac{\rho \lambda e^{-\rha(s-t)}}{2e^{X_{t}}} \EE^{*}_{t}\left(\frac{ e^{X_s}}{Y_s}\right),\\
\lim_{T\downarrow s}(T-s)\frac{\partial ^{2}I}{\partial \alpha ^{2}}(t,s;\alpha ^{*}) 
 & = \lambda^2\EE^{*}_{t}\left(\int_{t}^{s}\frac{e^{-2\kappa(s-u)}}{g(u,s)^{3}}du\right) 
 +\frac{1}{g(t,s)}
 - \frac{1}{g(t,s) + \rho \Ee_{t,s}}
 - \rho \EE^{*}_{t}\left(\frac{\Ee_{t,s}}{Y_s^{2}}\right).
\end{align*}

\appendix
\section{Proofs}\label{sec:Proofs}
\subsection{Proof of Theorem~\ref{the:hw}}\label{sec:the:hw}
Since the process~$M$ is a martingale, we can write
$$
V_{t} = e^{-\rha\left( T-t\right) }\EE^{*}_{t}\left( e^{X_{T}}-e^{\alpha}e^{X_{s}}\right)_{+}
 = e^{-\rha\left( T-t\right) }\EE^{*}_{t}\left( e^{X_{T}}-M_{T}\right)_{+}
 = e^{-\rha(T-t) }\BSB_T,
$$
where we denote $\BSB_T:=\BS(T,X_{T},M_{T},\nu_{T})$ for clarity in the next few pages.
The anticipating It\^{o}'s formula for the Skorohod integral~\cite[Theorem 3.2.2]{Nualart}~yields
\begin{align*}
 & \EE^{*}_{t}\left(e^{-\rha T}\BS(T,X_{T},M_{T},v_{T})\right) \\
 & = \EE^{*}_{t}\Big[ e^{-\rha t}\BSB_{t}-\rha
\int_{t}^{T}e^{-\rha u}\BSB_{u}du
 + \int_{t}^{T}e^{-\rha u}\partial_{u}\BSB_{u}du 
+\int_{t}^{T}e^{-\rha u}\partial_{\sigma}\BSB_{u}\left(\rha-\frac{\sigma_{u}^{2}}{2}\right) du \\
 & + \frac{1}{2}\int_{t}^{T}e^{-\rha u}\partial^{2}_{xx}\BSB_{u}\sigma_{u}^{2}du 
+\int_{t}^{T}e^{-\rha u}\partial^{2}_{xK}\BSB_{u}d\langle X,M\rangle_{u} \\
 & + \frac{1}{2}\int_{t}^{T}e^{-\rha u}\partial^{2}_{KK}\BSB_{u} d\langle M,M\rangle_{u}
 + \frac{1}{2}\int_{t}^{T}e^{-\rha u}\left(\Dxx-\Dx\right) \BSB_{u}
\left(v_{u}^{2}-\sigma_{u}^{2}\ind_{\left] s,T\right] }(u)\right) du \\
 & + \frac{1}{2}\int_{t}^{T}e^{-\rha u}\frac{\partial }{\partial x}\left( \Dxx-\Dx\right)
\BSB_{u}\sigma_{u}\rho \Lambda_{u}^{W^{*}}du \\
 & +\frac{1}{2}\int_{t}^{T}e^{-\rha u}\DK
\left(\Dxx-\Dx\right)\BSB_{u}\sigma_{u}e^{\alpha}e^{\rha
(s-u)}e^{X_{u}}\rho \Lambda_{u}^{W^{*}}\ind_{\left[ 0,s\right] }(u)du
\Big] .
\end{align*}
That is, since $t\leq s$, 
\begin{align*}
V_{t} & = \EE^{*}_{t}\Big[ \BS(t,X_{t},M_{t},v_{t})
+\int_{t}^{T}e^{-\rha\left( u-t\right) }\mathcal{L}_{\BS}\left(
v_{u}\right) \BSB_{u}du \\
 & + \frac{1}{2}\int_{t}^{T}e^{-\rha\left( u-t\right) }\left(\Dxx-\Dx\right)
\BSB_{u}\left(\sigma_{u}^{2}-v_{u}^{2}\right) du
 + \int_{t}^{T}e^{-\rha\left( u-t\right)}\partial^{2}_{xK}\BSB_{u}d\langle X,M\rangle_{u} \\
 & + \frac{1}{2}\int_{t}^{T}e^{-\rha\left( u-t\right) }\partial^{2}_{KK}\BSB_{u})d\langle M,M\rangle_{u}
  + \frac{1}{2}\int_{t}^{T}e^{-\rha\left( u-t\right) }\left(\Dxx-\Dx\right)
\BSB_{u}\left( v_{u}^{2}-\sigma_{u}^{2}\ind_{\left] s,T\right]}(u)\right) du \\
 & + \frac{1}{2}\int_{t}^{T}e^{-\rha\left( u-t\right) }\Dx\left(\Dxx-\Dx\right) \BSB_{u}\sigma_{u}\rho \Lambda
_{u}^{W^{*}}du \\
 & + \frac{1}{2}\int_{t}^{s}e^{-\rha(u-t) }
\DK\left(\Dxx-\Dx\right)\BSB_{u}\sigma_{u}e^{\alpha}
e^{\rha(s-u)}e^{X_{u}}\rho \Lambda_{u}^{W^{*}}du\Big] .
\end{align*}

Thus, we get, for $t\leq s$, 
\begin{align*}
V_{t} & = \EE^{*}_{t}\Big[\BSB_{t}
 + \int_{t}^{T}e^{-\rha\left( u-t\right) }\mathcal{L}_{\BS}\left(
v_{u}\right) \BSB_{u}du \\
 & + \frac{1}{2}\int_{t}^{T}e^{-\rha\left( u-t\right) }\left(\Dxx-\Dx\right)
\BSB_{u}\left( \sigma_{u}^{2}-\sigma_{u}^{2}\ind_{\left] s,T
\right] }(u)\right) du \\
 & + \int_{t}^{T}e^{-\rha\left( u-t\right) }\partial^{2}_{xK}\BSB_ud\langle X,M\rangle_{u}
 + \frac{1}{2}\int_{t}^{T}e^{-\rha(u-t)}\DKK\BSB_{u}d\langle M,M\rangle_{u} \\
 & + \frac{1}{2}\int_{t}^{T}e^{-\rha(u-t)}\Dx
\left(\Dxx-\Dx\right)\BSB_{u}\sigma_{u}\rho \Lambda_{u}^{W^{*}}du \\
 & +\frac{1}{2}\int_{t}^{s}e^{-\rha(u-t)}\DK
\left(\Dxx-\Dx\right) \BSB_{u}\sigma_{u}e^{\alpha}e^{\rha(s-u)}e^{X_{u}}\rho \Lambda_{u}^{W^{*}}du\Big] .
\end{align*}

Now, taking into account the identities $\mathcal{L}_{\BS}\left( v_{u}\right)\BSB_{u} = 0$ and 
\begin{equation*}
\begin{array}{rlcrl}
d\langle M,X\rangle_{u}
 & = \displaystyle \sigma_{u}^{2}e^{\alpha}e^{\rha(s-u)}e^{X_{u}}\ind_{\left[0,s\right] }(u)du,
& \displaystyle \partial^{2}_{xK}\BS
 & = & \displaystyle -\frac{1}{K}\left(\Dxx - \Dx\right)\BS,\\
d\langle M,M\rangle_{u}
& \displaystyle = \sigma_{u}^{2}{e^{2\alpha }}e^{2\rha(s-u)}e^{2X_{u}}\ind_{\left[ 0,s\right] }(u)du,
 & \displaystyle \DKK\BS & = & \displaystyle \frac{1}{K^{2}}\left(\Dxx-\Dx\right)\BS,
\end{array}
\end{equation*}
with $\BS$ evaluated at $(t,x,K,\sigma)$.
It follows that 
\begin{align*}
V_{t} & = \EE^{*}_{t}\Big[ \BS(t,X_{t},M_{t},v_{t})
 + \frac{1}{2}\int_{t}^{T}e^{-\rha\left( u-t\right) }\Dx\left(\Dxx-\Dx\right)\BSB_{u}\sigma_{u}\rho \Lambda
_{u}^{W^{*}}du \\
&+\frac{1}{2}\int_{t}^{s}e^{-\rha\left( u-t\right) }\DK\left(\Dxx-\Dx\right)\BSB_{u}\sigma_{u}e^{\alpha}
e^{\rha(s-u)}e^{X_{u}}\rho \Lambda_{u}^{W^{*}}du\Big] \\
&=\EE^{*}_{t}\Big[\BSB_{t}
+\frac{1}{2}\int_{s}^{T}e^{-\rha(u-t)}\Dx\left(\Dxx-\Dx\right)\BSB_{u}\sigma_{u}\rho \Lambda
_{u}^{W^{*}}du \\
&+\frac{1}{2}\int_{t}^{s}e^{-\rha\left( u-t\right) }\Dx\left(\Dxx-\Dx\right)
\BSB_{u}\sigma_{u}\rho \Lambda_{u}^{W^{*}}du \\
& + \frac{1}{2}\int_{t}^{s}e^{-\rha\left( u-t\right) }\DK\left(\Dxx-\Dx\right) \BSB_{u}\sigma_{u}M_{u}\rho
\Lambda_{u}^{W^{*}}du\Big] .
\end{align*}
Since the Black-Scholes function satisfies 
\begin{align*}
\Dx\left(\Dxx-\Dx\right)\BS\left( t,x,k,\sigma \right)
 & = \frac{e^{x}\Nn'(d_{+})}{\sigma \sqrt{T-t}}\left( 1-\frac{d_{+}}{\sigma 
\sqrt{T-t}}\right),\\
\DK\left(\Dxx-\Dx\right)\BS\left( t,x,k,\sigma \right)
 & = \frac{e^{x}\Nn'(d_{+})}{K\sigma \sqrt{T-t}}\left( \frac{d_{+}}{\sigma \sqrt{T-t}}\right) ,
\end{align*}
then it is easy to see that 
\begin{align*}
V_{t} & = \EE^{*}_{t}\Big[\BSB_{t}
+\frac{\rho }{2}\int_{s}^{T}e^{-\rha\left( u-t\right)
}H(u,X_{u},M_{u},v_{u})\sigma_{u}\Lambda_{u}^{W^{*}}du \\
& + \frac{\rho }{2}\int_{t}^{s}e^{-\rha\left( u-t\right)
}G(u,X_{u},M_{u},v_{u})\sigma_{u}\Lambda_{u}^{W^{*}}du\Big].
\end{align*}
The proof then follows from the easy computations
\begin{align*}
\BSB_{t}
 & = e^{X_{t}}\Nn\left( \frac{-{\alpha }+\rha(T-s)}{v_{s}\sqrt{T-s}}+\frac{v_{s}\sqrt{T-s}}{2}\right) \\
 & -e^{\alpha}e^{X_{t}}e^{-\rha(T-s)}\Nn\left( \frac{{-\alpha }+\rha(T-s)}{v_{s}\sqrt{T-s}}-\frac{v_{s}\sqrt{T-s}}{2}\right)
= e^{X_t}\BS\left( s,0,e^{\alpha},v_{s}\right)
\end{align*}
and, for all $u<s$,
$$
G(u,X_{u},M_{u},v_{u}) = \frac{e^{X_{u}}}{v_{s}\sqrt{T-s}}
\Nn'\left( \frac{{-\alpha }+
\rha(T-s)}{v_{s}\sqrt{T-s}}+\frac{v_{s}\sqrt{T-s}}{2}\right)
 = e^{X_{u}}G(s,0,e^{\alpha},v_{s}).
$$

\subsection{Proof of Theorem~\ref{el13}}\label{sec:el13}
Before diving into the proof, let us state and prove the statement when the correlation~$\rho$ is null.
We shall always consider $t<s$.
\begin{lemma}\label{lem:limimp}
Under~\textbf{(H1)}-\textbf{(H2)}, if $\rho =0$, then 
$\lim\limits_{T\downarrow s}I(t,s;\alpha ^{*})=\EE^{*}_{t}\left( \sigma_{s}\right)$.
\end{lemma}
\begin{proof}
From the definition~\eqref{elprecio} of the implied volatility, we can write
\begin{align*}
I(t,s;\alpha ^{*})
 & = \BS^{-1}\circ\EE^{*}_{t}\left[\BS\left( s,0,e^{\alpha^{*}},v_{s}\right)\right]\\
 & = \EE^{*}_{t}\Big\{ \BS^{-1}\circ\BS\left( s,0,e^{\alpha^{*}},v_{s}\right)\\
 & \quad +\BS^{-1}\circ\EE^{*}_{t}\left[\BS\left( s,0,e^{\alpha^{*}},v_{s}\right)
\right]
 - \BS^{-1}\circ\BS\left( s,0,e^{\alpha^{*}},v_{s}\right)\Big\}\\
 & = \EE^{*}_{t}\left[ v_{s}+\BS^{-1}\circ\EE^{*}_{t}\left[\BS\left(s,0,e^{\alpha^{*}},v_{s}\right)\right] -\BS^{-1}\circ\BS\left( s,0,e^{\alpha^{*}},v_{s}\right)\right].
\end{align*}
A direct application of Clark-Ocone's formula~\cite[Proposition 1.3.14]{Nualart} yields
$\BS\left(s,0,e^{\alpha^{*}},v_{s}\right) = A_T$,
with
\begin{align}
U_{u} & := \EE^{*}_{u}\left[\frac{\partial \BS}{\partial \sigma }\left( s,0,e^{\alpha^{*}},v_{s}\right) \frac{D_{u}^{W^{*}}\int_{s}^{T}\sigma_{r}^{2}dr}{2(T-s)v_{s}}\right],\label{eq1}\\
A_u & := \EE^{*}_{t}\left[\BS(s,0,e^{\alpha^*},v_{s})\right] +\int_{t}^{u}U_{r}dW_{r}.\nonumber
\end{align}
Applying It\^{o}'s formula to 
$\BS^{-1}(A_u)$ and taking expectations, we obtain
\begin{align*}
 & \EE^{*}_{t}\Big[\BS^{-1}\circ\EE^{*}_{t}\left[\BS\left(s,0, e^{\alpha^{*}},v_{s}\right]\right)
  - \BS^{-1}\circ\BS\left(s,0,e^{\alpha^{*}},v_{s}\right)\Big]\\
& = -\frac{1}{2}\EE^{*}_{t}\left\{\int_{t}^{T}(\BS^{-1})''
\left(\EE_{r}\left[\BS\left(s,0,e^{\alpha^{*}},v_{s}\right)\right]\right)U_{r}^{2}dr\right\},
\end{align*}
so that
$$
\lim_{T\downarrow s}I(t,s,\alpha ^{*})
 = \EE^{*}_{t}\left( \sigma_{s}\right) 
 - \frac{1}{2}\lim_{T\downarrow s}
\EE^{*}_{t}\left\{\int_{t}^{T}(\BS^{-1})''
\left(\EE_{r}\left[\BS\left( s,0,e^{\alpha^{*}},v_{s}\right)\right]\right)U_{r}^{2}dr\right\}.
$$
Now, considering that 
$$
\left(\BS^{-1}\right)''\left(\EE_{r}\left[\BS\left( s,0,e^{\alpha^*},v_{s}\right)\right]\right)
 = \frac{\BS^{-1}\left(\EE_{r}\left[\BS\left( s,0,e^{\alpha^*},v_{s}\right)\right](T-s)\right)}{4\Nn'(\gamma)^2(T-s)},
$$
where 
$\gamma:=\frac{1}{2}\BS^{-1}\left(\EE_{r}\left[\BS\left(s,0,e^{\alpha^{*}},v_{s}\right)\right]\sqrt{T-s}\right)$, 
the lemma follows from
$$
\lim_{T\downarrow s}\frac{1}{2}\EE^{*}_{t}\left\{\int_{t}^{T}(\BS^{-1})''
\left(\EE_{r}\left[\BS\left( s,0,e^{\alpha^*},v_{s}\right)\right]\right)U_{r}^{2}dr\right\}=0.
$$
\end{proof}

The following technical lemma, 
which follows similar arguments to~\cite[Lemma 4.1]{alv07},
shall be of fundamental importance in the proof of the theorem:

\begin{lemma}\label{lemaconJosep}
Let $0\leq t\leq s, u < T$. 
For every $n\geq 0$, there exists $C>0$ such that 
$$
\left|\EE\left( \left. {\partial_{x}^{n}G}\left( u,X_{u}, M_u,v_{u}\right)\right| \mathcal{F}_{t}\vee\mathcal{F}_{T}^{W^*}\right) \right|
 \leq 
C \EE\left. \left( e^{X_{s\wedge u}} \right| \mathcal{F}_{t}\vee\mathcal{F}_{T}^{W^*}\right) 
\left( \int_{s}^{T}\sigma_{s}^{2}ds\right) ^{-\frac{1}{2}\left( n+1\right) }.
$$
\end{lemma}

\begin{proof}[Proof of Theorem~\ref{el13}]
We know, from Theorem~\ref{the:hw}, that 
\begin{align*}
I(t,s;\alpha^{*})
 & = \BS^{-1}\Big\{ \EE^{*}_{t}\left[\BS\left( s,0,e^{\alpha^*},v_{s}\right)\right]\\
& + \frac{\rho }{2}e^{-X_{t}}\EE^{*}_{t}\left(\int_{s}^{T}e^{-\rha\left( u-t\right)
}H(u,X_{u},M_{u},v_{u})\sigma_{u}\Lambda_{u}^{W^{*}}du\right)\\
& + \frac{\rho }{2}e^{-X_{t}}\EE^{*}_{t}\left( G(s,0,e^{\alpha^*},v_{s})\int_{t}^{s}e^{-\rha(u-t)}e^{X_{u}}\sigma_{u}\Lambda
_{u}^{W^{*}}du\right) \Big\}.
\end{align*}
By the mean value theorem, we can find $\theta$ between
$\EE^{*}_{t}\left[\BS( s,0,e^{\alpha^*},v_{s})\right]$
and
\begin{align*}
 & \EE^{*}_{t}\left[\BS\left( s,0,e^{\alpha^*},v_{s}\right)
 + \frac{\rho }{2}e^{-X_{t}}\int_{s}^{T}e^{-\rha\left( u-t\right)
}H(u,X_{u},M_{u},v_{u})\sigma_{u}\Lambda_{u}^{W^{*}}du \right.\\
& +\left.\frac{\rho }{2}e^{-X_{t}} G(s,0,e^{\alpha^*},v_{s})\int_{t}^{s}e^{-\rha(u-t)}e^{X_{u}}\sigma_{u}\Lambda
_{u}^{W^{*}}du\right],
\end{align*}
such that, denoting $\widetilde{I}_1, \widetilde{I}_2$ respectively the second and third expectations,
$$
I(t,s;\alpha ^{*}) - \BS^{-1}\circ\EE^{*}_{t}\left[\BS( s,0,e^{\alpha^*},v_{s})\right]
 = \frac{\sqrt{2\pi}\exp\left\{\frac{\BS^{-1}(\theta)}{8}(T-s)\right\}}{\sqrt{T-s}}
(\widetilde{I}_1+\widetilde{I}_2)
 =: I_1 + I_2.
$$
From Lemma~\ref{lemaconJosep}, we have
\begin{align*}
\lim_{T\downarrow s}\left|I_1\right|
 & \leq C\lim_{T\downarrow s}\frac{e^{- X_t}}{\sqrt{T-s}}\int_s^T
 \frac{\EE^{*}\left(e^{X_s}\vert\mathcal{G}_t\right)}{T-s}\left|\Lambda_{u}^{W^{*}} \right|du
 \leq Ce^{- X_t}\lim_{T\downarrow s}
 \sqrt{(T-s)\EE^{*}\left(e^{2X_s}\left|\mathcal{G}_t\right.\right)}=0,
\end{align*}
which, together with Lemma~\ref{lem:limimp}, implies
\begin{align*}
\lim_{T\downarrow s}I(t,s;\alpha ^{*})
 & = \EE^{*}_{t}(\sigma_s) + \frac{\rho e^{X_t}}{2}
\lim_{T\downarrow s}\EE^{*}_{t}\left(\frac{\exp\left(\frac{v_s^2(T-s)}{8}\right)}{v_s(T-s)}
\int_{t}^{s}e^{-\rha(u-t)}e^{X_{u}}\sigma_{u}\Lambda
_{u}^{W^{*}}du
\right)\\
 & = \EE^{*}_{t}(\sigma_s)
+\frac{\rho e^{X_t}}{2}\EE^{*}_{t}\left(\frac{1}{\sigma_s}
\int_{t}^{s}e^{-\rha(u-t)}e^{X_{u}}\sigma_{u}\left(D_{u}^{W^{*}}\sigma^2_s\right)du
\right),
\end{align*}
and the theorem follows.
\end{proof}

\subsection{Proof of Theorem~\ref{the:pdi}}\label{sec:the:pdi}
Differentiating~\eqref{elprecio} with respect to the log-forward moneyness~$\alpha$ yields
$$
\frac{\partial V_{t}}{\partial \alpha}
 = e^{X_{t}}\Dk\BS\left( s,0,e^{\alpha},I(t,s;\alpha )\right) \\
 + e^{X_{t}}\Ds\BS\left( s,0,e^{\alpha},I(t,s;\alpha )\right) \frac{\partial I}{\partial \alpha }(t,s;\alpha),
$$
with $k:=\log(K)$.
Then, from Theorem~\ref{the:hw} we are able to write 
\begin{equation}\label{eq:DecompSkew}
\begin{array}{rl}
\displaystyle \frac{\partial I}{\partial \alpha}(t,s;\alpha) 
 & = \displaystyle \frac{\frac{\partial V_{t}}{\partial \alpha }-e^{X_{t}}\Dk\BS\left( s,0,e^{\alpha},I(t,s;\alpha )\right)}
{e^{X_{t}}\Ds\BS\left( s,0,e^{\alpha},I(t,s;\alpha)\right)}  \\
\text{} & = \displaystyle \frac{\EE^{*}_{t}\left[\Dk\BS\left( s,0,{e^{\alpha }},{v}_{s}\right) \right]
 - \Dk\BS\left(s,0,e^{\alpha},I(t,s;\alpha )\right) }{\Ds\BS\left( s,0,e^{\alpha},I(t,s;\alpha )\right)}  \\
\text{} & \displaystyle + \frac{\rho}{2}\frac{
\EE^{*}_{t}\left[\int_{s}^{T}e^{-\rha(u-t) }\Dk H(u,X_{u},e^{\alpha}
e^{X_{s}},v_{u})\sigma_{u}\Lambda_{u}^{W^{*}}du\right]}{e^{X_{t}}
\Ds\BS\left(s,0,e^{\alpha},I(t,s;\alpha)\right)}\\
\text{} & \displaystyle  + \frac{\rho}{2}\frac{\EE^{*}_{t}\left[\Dk G(s,0,e^{\alpha},{v}_{s})\int_{t}^{s}e^{X_{u}}{e^{-\rha(u-t)}}\sigma_{u}\Lambda_{u}^{W^{*}}du\right] }{e^{X_{t}}\Ds\BS
\left(s,0,e^{\alpha},I(t,s;\alpha)\right)}
 =: T_{1}+T_{2}+T_{3}.
\end{array}
\end{equation}
In the uncorrelated case $\rho =0$, this expression simplifies to 
$$
\frac{\partial I}{\partial \alpha }(t,s;\alpha )
 = \frac{\EE^{*}_{t}\left[\Dk\BS\left( s,0,e^{\alpha},{v}_{s}\right)\right]
 - \Dk\BS \left(s,0,e^{\alpha},I(t,s;\alpha )\right)}{\Ds\BS \left( s,0,e^{\alpha},I(t,s;\alpha)\right)},
$$
and in the at-the-money case $\alpha =\alpha^{*}$, we obtain, by Theorem~\ref{the:hw}, 
\begin{align*}
\EE^{*}_{t}\left[\Dk\BS\left( s,0,e^{\alpha ^{*}},v_{s}\right) \right]
 & = \EE^{*}_{t}\left[ -\Nn\left( -\frac{{v}_{s}\sqrt{T-s}}{2}\right) \right]
 = \EE^{*}_{t}\left[ \frac{\Nn\left( \frac{{v}_{s}\sqrt{T-s}}{2}\right)
-\Nn\left( -\frac{{v}_{s}\sqrt{T-s}}{2}\right) -1}{2}\right] \\
 & =\EE^{*}_{t}\left[ \frac{\BS\left(s,0,e^{\alpha^*},{v}_{s}\right) -1}{2}\right]
  = \frac{V_{t}e^{-X_{t}} -1}{2}.
\end{align*}
Since furthermore
$$
\Dk\BS\left( s,0,e^{\alpha^{*}},I(t,s;\alpha ^{*})\right)
 = - \Nn\left( -\frac{I(t,s;\alpha ^{*})\sqrt{T-s}}{2}\right)
 = \frac{V_{t}e^{-X_{t}} -1}{2},
 $$
then the at-the-money forward skew $\frac{\partial I}{\partial \alpha }(t,s;\alpha^{*})$ is null.
Proceeding now to the general correlated case, the decomposition~\eqref{eq:DecompSkew}
and Theorem~\ref{the:hw} yield
\begin{align*}
&\EE^{*}_{t}\left[ \Dk\BS(s,0,e^{\alpha^{*}},v_{s})\right]
 - \Dk\BS\left( s,0,e^{\alpha^*},I(t,s;\alpha ^{*})\right) \\
&=\EE^{*}_{t}\left[ \frac{\BS\left( s,0,e^{\alpha^*},v_{s}\right) -1}{2}\right] -\frac{V_{t}e^{-X_{t}} -1}{2} \\
&=-\frac{e^{-X_{t}}}{2}\EE^{*}_{t}\left\{ \frac{\rho }{2}
\int_{s}^{T}e^{-\rha(u-t)}H(u,X_{u},M_{u},v_{u})\sigma_{u}\Lambda_{u}^{W^{\ast
}}du|_{\alpha =\alpha ^{*}}\right. \\
&\left. -\frac{e^{-X_{t}}\rho }{4}G(s,0,e^{\alpha^{*}},v_{s})\int_{t}^{s}e^{-\rha\left( u-t\right) }\sigma
_{u}e^{X_{u}}\Lambda_{u}^{W^{*}}du\right\} .
\end{align*}
Then, 
$\lim\limits_{T\downarrow s}T_{1} = \lim\limits_{T\downarrow s}I_{1} + \lim\limits_{T\downarrow s}I_{2}$,
where
\begin{align*}
I_{1} & = -\frac{e^{-X_{t}}\rho \EE^{*}_{t}\left(
\int_{s}^{T}e^{-\rha\left( u-t\right) }H(u,X_{u},M_{u},v_{u})\sigma
_{u}\Lambda_{u}^{W^{*}}du\right) |_{\alpha =\alpha ^{*}}}{4\frac{
\partial \BS}{\partial \sigma }\left( s,0,e^{\alpha^*},I(t,s;\alpha ^{*})\right) },\\
I_{2} & = -\frac{e^{-X_{t}}\rho \EE^{*}_{t}\left( G(s,0,e^{\alpha ^{*}},v_{s})\int_{t}^{s}e^{-\rha\left( u-t\right) }\sigma
_{u}e^{X_{u}}\Lambda_{u}^{W^{*}}du\right) }{4\frac{\partial \BS}{\partial \sigma }\left( s,0,e^{\alpha^*},I(t,s;\alpha ^{\ast
})\right) }.
\end{align*}
Under~\textbf{(H3)}, $\lim\limits_{T\downarrow s}I_{1}=0$ by Lemma~\ref{lemaconJosep}. 
Furthermore 
$G(s,0,e^{\alpha^*},{v}_{s})=2\Dk G(s,0,e^{\alpha^{*}},v_{s})$,
so that
$$
I_{2}
 = -\frac{\rho }{2}\frac{e^{-X_{t}}\EE^{*}_{t}\left( \Dk G(s,0,e^{\alpha^*},v_{s})
\int_{t}^{s}e^{-\rha\left( u-t\right) }e^{X_{u}}\sigma_{u}\Lambda
_{u}^{W^{*}}du\right)}
{\Ds\BS\left(s,0,e^{\alpha^*},I(t,s;\alpha ^{*})\right) }
 = -T_{3}.
$$
On the other hand, using Theorem~\ref{the:pdi}, Lemma~\ref{lemaconJosep} and the
anticipating It\^{o}'s formula again, it follows that 
\begin{align*}
\lim_{T\downarrow s} T_{2}
 & = \lim_{T\downarrow s}\frac{\frac{\rho }{2}\EE^{*}_{t}\left[
\int_{s}^{T}e^{-\rha(u-t)}\Dk H(u,X_{u},M_{u},v_{u})\sigma_{u}
\Lambda_{u}^{W^{*}}du\right]}{e^{X_{t}}\Ds\BS\left( s,0,e^{\alpha^*},I(t,s;\alpha ^{*})\right)}\\
 & = \frac{\rho }{2}\lim_{T\downarrow s}\frac{\EE^{*}_{t}\left[ \Dk H(s,X_{s},M_{s},v_{s})e^{-\rha\left( s-t\right)}
 \int_{s}^{T}\sigma_{u}\Lambda_{u}^{W^{*}}du\right] }{e^{X_{t}}
\Ds\BS\left( s,0,{e^{\alpha^*}},I(t,s;\alpha ^{*})\right) } \\
 & = \frac{\rho }{2}\lim_{T\downarrow s}\EE^{*}_{t}\left[ \frac{e^{X_{s}}e^{-\rha(s-t)}}{e^{X_{t}}v_{s}^{3}(T-s)^{2}}\int_{s}^{T}\sigma
_{u}\Lambda_{u}^{W^{*}}du\right] \\
 & = \frac{\rho}{2e^{X_{t}}}\left\{\lim_{T\downarrow s}
 \EE^{*}_{t}\left(
\frac{e^{X_{s}}e^{-\rha(s-t)}}{\sigma_{s}^{3}(T-s)^{2}}
\int_{s}^{T}\sigma_{u}\Lambda_{u}^{W^{*}}du\right)\right\}^{-1}.
\end{align*}
Now, \eqref{condition} allows us to write 
\begin{equation*}
\lim_{T\downarrow s}T_{2}
 = \frac{\rho e^{-\rha(s-t)}}{4e^{X_{t}}}
 \EE^{*}_{t}\left(\frac{\overline{\sigma}_{s}^{2}}{\sigma_{s}^{2}}\right)^{-1},
\end{equation*}
which completes the proof of the theorem.


\subsection{Proof of Theorem~\ref{el14}}\label{sec:el14}
The proof of the theorem will be split into three parts:
we first state and prove the technical lemma below, similar to~\cite[Theorem 5]{Alos2015},
which we then use to prove the theorem in the uncorrelated case $\rho=0$, 
and we finally proceed to the general case.

\subsubsection{A technical lemma}
\begin{lemma}\label{lem:aux8} 
Under~\textbf{(H1)}-\textbf{(H2)}, the equality 
$$
\Ds\BS\left( s,0,e^{\alpha^*},I(t,s;\alpha ^{*})\right) \frac{\partial ^{2}I}{\partial \alpha ^{2}}
(t,s;\alpha ^{*})
 =  \frac{1}{2}\EE^{*}_{t}\left[ \int_{t}^{T}\frac{\partial ^{2}\Psi }{
\partial a^{2}}\left( \EE^{*}_{u}\left( \BS\left( s,0,e^{\alpha^*},v_{s}\right) \right) \right) U_{u}^{2}du\right]
$$
holds, where 
$\Psi (a):=
\displaystyle \Dkk\BS\left( s,0,e^{\alpha^{*}},\BS^{-1}\left( a\right) \right)$
and~$U$ is given in~\eqref{eq1}.
\end{lemma}

\begin{proof}
This proof is similar to that in~\cite[Theorem 5]{Alos2015}, so we only sketch it. 
Differentiating the option price with respect to the log-moneyness yields
$$
\frac{\partial V_{t}}{\partial \alpha} = e^{X_{t}}\Dk\BS\left( s,0,e^{\alpha},I(t,s;\alpha )\right)
+e^{X_{t}}\Ds\BS\left( s,0,e^{\alpha},I(t,s;\alpha )\right) \frac{\partial I}{\partial \alpha }(t,s;\alpha),
$$
and
\begin{align}
\frac{\partial ^{2}V_{t}}{\partial \alpha ^{2}}
 & = e^{X_{t}}\Dkk\BS\left( s,0,e^{\alpha},I(t,s;\alpha
)\right)
 +{2}e^{X_{t}}\Dsk\BS\left(
s,0,e^{\alpha},I(t,s;\alpha )\right) \frac{\partial I}{\partial \alpha 
}(t,s;\alpha ) \nonumber \\
& +e^{X_{t}}\Dss\BS\left( s,0,e^{\alpha},I(t,s;\alpha )\right) \frac{\partial I}{\partial \alpha}(t,s;\alpha)^{2}\nonumber \\
 &+e^{X_{t}}\Ds\BS\left( s,0,e^{\alpha},I(t,s;\alpha )\right) \frac{\partial ^{2}I}{\partial \alpha ^{2}}
(t,s;\alpha ).\label{eq:sdi}
\end{align}
Combining~\eqref{eq:sdi} and~\eqref{TheBScase} yields 
$$
\left.\frac{\partial ^{2}V_{t}}{\partial \alpha ^{2}}\right|_{\alpha =\alpha^{*}}
 = e^{X_{t}}\Dkk\BS\left( s,0,e^{\alpha ^{*}},I(t,s;\alpha ^{*})\right)
 + e^{X_{t}}\Ds\BS\left( s,0,e^{\alpha^{*}},I(t,s;\alpha ^{*})\right) \frac{\partial ^{2}I}{\partial
\alpha ^{2}}(t,s;\alpha ^{*}).
$$
Then, taking into account Theorem~\ref{the:hw} and the fact that 
$I(t,s;\alpha ^{*})=\BS^{-1}(e^{-X_{t}}V_{t})$, we are able to write 
\begin{align*}
 & e^{-X_{t}}\Ds\BS\left( s,0,e^{\alpha^{*}},I(t,s;\alpha ^{*})\right) \frac{\partial ^{2}I}{
\partial \alpha ^{2}}(t,s;\alpha ^{*}) \\
 & = \left.\frac{\partial ^{2}V_{t}}{\partial \alpha ^{2}}\right|_{\alpha =\alpha ^{\ast}}
  - e^{-X_{t}}\Dkk\BS\left( s,0,e^{\alpha ^{*}},I(t,s;\alpha ^{*})\right) \\
 & = e^{-X_{t}}\EE^{*}_{t}\left(\Dkk\BS\left(
s,0,e^{\alpha^*},v_{s}\right) - \Dkk\BS\left( s,0,{e^{\alpha^*}},I(t,s;\alpha ^{*})\right)
\right) \\
 & = e^{X_{t}}\EE^{*}_{t}\left[ \Dkk\BS\left(
s,0,{e^{\alpha^*}},\BS^{-1}\circ\BS(s,0,{e^{\alpha^*}},v_{s})\right) \right. \\
 & \left. - \Dkk\BS\left( s,0,e^{\alpha^{*}},\BS^{-1}\circ\EE^{*}_{t}\left( \BS(s,0,{e^{\alpha^*}},v_{s}) \right) \right)\right] .
\end{align*}
Clark-Ocone's formula above and It\^{o}'s formula applied to the process 
\begin{equation*}
A_{u}:=\Dkk\BS\left( s,0,e^{\alpha^{*}},\BS^{-1}\circ\EE^{*}_{u}\left( \BS(s,0,{e^{\alpha^*}},v_{s}\right) \right)
\end{equation*}
imply the result after taking expectations.
\end{proof}

\subsubsection{The uncorrelated case}
We now move on to the proof of Theorem~\ref{el14} in the uncorrelated case $\rho=0$.
Lemma~\ref{lem:aux8} yields
\begin{align*}
 & \lim_{T\downarrow s}(T-s)\frac{\partial ^{2}I}{\partial \alpha ^{2}}(t,s;\alpha ^{*}) \\
 & = \displaystyle \lim_{T\downarrow s}(T-s)\left( \frac{\EE^{*}_{t}\left[
\int_{t}^{s}\partial^{2}_{aa}\Psi\left(\EE^{*}_{u}\left(
\BS\left( s,0,e^{\alpha^*},v_{s}\right) \right) \right) U_{u}^{2}du
\right] }{2\Ds\BS\left( s,0,{e^{\alpha^{*}}},I(t,s;\alpha ^{*})\right) }\right)  \\
 &  \displaystyle + \lim_{T\downarrow s}(T-s)\left( \frac{\EE^{*}_{t}\left[
\int_{s}^{T}\partial^{2}_{aa}\Psi\left(\EE^{*}_{u}\left(
\BS\left( s,0,\alpha ^{*},v_{s}\right) \right) \right) U_{u}^{2}du\right] 
}{2\Ds \BS\left( s,0,{e^{\alpha^*}},I(t,s;\alpha ^{*})\right) }\right)
=: \lim_{T\downarrow s}(T_{1}+T_{2}).
\end{align*}
Since by Assumption~\textbf{(H1)} the process~$\sigma$ is bounded below and above almost surely, 
using that $\BS(s,0,e^{\rha(T-s)},\cdot )$ is an increasing function, together with~\textbf{(H3)} and
$$
\frac{\partial^{2}\Psi}{\partial a^{2}}\left(\EE^{*}_{u}\left(\BS\left( s,0,\alpha ,v_{s}\right) \right)\right)
 = \frac{2\sqrt{2\pi}}{(T-s)^{3/2}}\frac{\exp\left\{\frac{T-s}{8}\left(\BS^{-1}\circ \EE^*_u(\BS(s,0,e^{\rha(T-s)},v_s))\right)^2\right\}}{
 \left(\BS^{-1}\circ\EE^*_u(\BS(s,0,e^{\rha(T-s)},v_s))\right)^3},
$$
we deduce that
$$
0 < T_{2}\leq C\EE^{*}_{t}\left( \int_{s}^{T}\frac{\exp{\left(\frac{C(T-s)}{8}\right)}}{{c^{3}(T-s)}}U_{u}^{2}du\right)
 \leq \frac{C}{T-s}\int_{s}^{T}\EE^{*}_{t}\left( U_{u}^{2}\right) du
 \leq C \sqrt{T-s},
$$
where $C>0$ is a constant that may change from line to line,
and hence $\lim_{T\downarrow s}T_{2}=0$.
Finally, Dominated Convergence Theorem and Lemma~\ref{lem:limimp} yield
\begin{align*}
\lim_{T\downarrow s}T_{1}
 & = \pi \lim_{T\rightarrow
s}\int_{t}^{s}\frac{\EE^{*}_{t}\left( U_{u}^{2}\right)}{I(u,s;\alpha ^{*})^{3}(T-s)}
\exp\left\{\frac{I(u,s;\alpha ^{*})^{2}(T-s)}{8}\right\} du \\
 & = \frac{1}{4}\lim_{T\downarrow s}\EE^{*}_{t}\int_{t}^{s}\frac{1}{I(u,s;\alpha ^{*})^{3}(T-s)}
\EE^{*}_{u}\left( \frac{\int_{s}^{T}(D_{u}^{W^{*}}
\sigma_{r}^{2})dr}{v_{s}\sqrt{T-s}}\right)^{2}du \\
 & = \frac{1}{4}\EE^{*}_{t}\left( \int_{t}^{s}
 \EE^{*}_{u}\left( \frac{D_{s}^{W*}\sigma_{r}^{2}}{\sigma_{s}}\right)^{2}
 \frac{du}{\EE^{*}_{u}(\sigma_{s})^{3}}\right),
\end{align*}
which concludes the proof.

\subsubsection{The general case}
Using Theorem~\ref{the:hw} and~\eqref{eq:sdi}, we can write
\begin{align*}
 & e^{X_{t}}\Ds\BS\left( s,0,e^{\alpha^{*}},I(t,s;\alpha^* )\right) \frac{\partial ^{2}I}{\partial \alpha ^{2}}(t,s;\alpha^* )\\
 & = e^{X_{t}}\EE^{*}_{t}\left[\Dkk\BS
\left( s,0,{e^{\alpha^*}},v_{s}\right) - \Dkk\BS\left( s,0,e^{\alpha^*},I(t,s;\alpha^* )\right)
\right]\\
& -2e^{X_{t}}\partial^{2}_{\sigma k}\BS\left( s,0,
{e^{\alpha^*}},I(t,s;\alpha^* )\right) \frac{\partial I}{\partial
\alpha }(t,s;\alpha^*)
 -e^{X_{t}}\Dss\BS\left( s,0,e^{\alpha^*},I(t,s;\alpha ^{*})\right)
 \left[\frac{\partial I}{\partial \alpha }(t,s;\alpha ^{*})\right]^{2}\\
& +\frac{\rho }{2}\EE^{*}_{t}\left. \left[
\int_{s}^{T}e^{-\rha(u-t)}\Dkk H(u,X_{u},M_{u},v_{u})\sigma_{u}\Lambda_{u}^{W^{*}}du\right. \right\vert_{\alpha ={\alpha ^{*}}}\\
& \left. + \Dkk G(s,0,{e^{\alpha^*}},v_{s})
\int_{t}^{s}e^{-\rha(u-t)}e^{X_{u}}\sigma_{u}\Lambda_{u}^{W^{*}}du\right].
\end{align*}

Therefore
\begin{align*}
 & (T-s)\partial^{2}_{\alpha\alpha}I(t,s;\alpha^{*})\\
 & = (T-s)\frac{\EE^{*}_{t}\left(\Dkk\BS\left( s,0,{e^{\alpha^*}},v_{s}\right)
 - \Dkk\BS\left( s,0,e^{\alpha^*},I^{0}(t,s;\alpha ^{*})\right)\right)}{\Ds\BS\left( s,0,e^{\alpha^*},I(t,s;\alpha ^{*})\right) } \\
 & + (T-s)\frac{\EE^{*}_{t}\left(\Dkk\BS
\left( s,0,{e^{\alpha^*}},I^{0}(t,s;\alpha ^{*})\right) - \Dkk\BS\left( s,0,{e^{\alpha^*}}
,I(t,s;\alpha ^{*})\right) \right) }{\Ds\BS\left( s,0,{e^{\alpha^*}},I(t,s;\alpha ^{*})\right) } \\
 & - 2(T-s)\frac{\partial^{2}_{\sigma k}\BS\left( s,0,{
e^{\alpha^*}},I(t,s;\alpha ^{*})\right) \frac{\partial I}{
\partial \alpha }(t,s;\alpha ^{*})}{\Ds\BS
\left( s,0,{e^{\alpha^*}},I(t,s;\alpha ^{*})\right) }\\
& - (T-s)\frac{\Dss\BS\left( s,0,{e^{\alpha^*}},I(t,s;\alpha ^{*})\right) \left(\partial_{\alpha}I(t,s;\alpha ^{*})\right) ^{2}}{\Ds\BS\left( s,0,{e^{\alpha^*}},I(t,s;\alpha ^{\ast
})\right) }\\
&+ \frac{\rho(T-s)}{2e^{X_{t}}}\left. \frac{\EE^{*}_{t}\left[
\int_{s}^{T}e^{-\rha\left( u-t\right)}\Dkk H(u,X_{u},M_{u},v_{u})\sigma_{u}\Lambda_{u}^{W^{*}}du\right. }{
\Ds\BS\left( s,0,{e^{\alpha^*}}
,I(t,s;\alpha ^{*})\right) }\right\vert_{\alpha =\alpha ^{*}} \\
& + \frac{\rho(T-s)}{2e^{X_{t}}}\frac{\EE^{*}_{t}\left(\Dkk G(s,0,{e^{\alpha^*}}
,v_{s})\int_{t}^{s}e^{-\rha(u-t)}e^{X_{u}}\sigma_{u}\Lambda
_{u}^{W^{*}} du\right)}{\Ds\BS\left( s,0,e^{\alpha^*},I(t,s;\alpha ^{*})\right)}\\
& =: T_{1}+T_{2}+T_{3}+T_{4}+T_{5}+T_{6}.
\end{align*}
Here $I^{0}(t,s;\alpha )$ denotes the forward implied volatility in the
uncorrelated case.
By definition of the Black-Scholes function and using Theorem~\ref{the:pdi}, we have
$\lim_{T\downarrow s}\left( T_{3}+T_{4}+T_{5}\right) \ =0$, and 
$$
\lim_{T\downarrow s}\frac{\partial_{\sigma}\BS\left( s,0,{
e^{\alpha^*}},I(t,s;\alpha ^{*})\right)}{\partial_{\sigma}\BS\left( s,0,{e^{\alpha^*}},I^{0}(t,s;\alpha ^{\ast
})\right) }=1.
$$
Then, the proof of Lemma~\ref{lem:aux8} yields
$\lim_{T\downarrow s}T_{1}=\lim_{T\downarrow s}(T-s)\partial^{2}_{kk}I^{0}(t,s,\alpha ^{*})$.
Computing the second derivative $\partial^2_{kk}\BS$, we also obtain
$$
\lim_{T\downarrow s}T_{2}=\lim_{T\downarrow s}\left( \frac{1}{
I^{0}(t,s;\alpha ^{*})}-\frac{1}{I(t,s;\alpha ^{*})}\right).
$$
On the other hand,
\begin{align*}
\lim_{T\downarrow s}T_{6}
 & = \frac{\rho }{2}e^{-X_{t}}\lim_{T\downarrow s}{\ }(T-s)\frac{\EE^{*}_{t}\left(\partial^{2}_{kk}G(s,0,{e^{\alpha^*}},v_{s})\int_{t}^{s}e^{-\rha
(u-t)}e^{X_{u}}\sigma_{u}\Lambda_{u}^{W^{*}}du\right)}{\partial_{\sigma}\BS\left( s,0,{e^{\alpha^*}},I(t,s;\alpha ^{\ast})\right) } \\
 & = -\frac{\rho }{2}e^{-X_{t}}\EE^{*}_{t}\left( \frac{1}{\sigma_{s}^{3}}
\int_{t}^{s}e^{-\rha(u-t)}e^{X_{u}}\sigma_{u}\left( D_{u}^{W^{*}}\sigma
_{s}^{2}\right) du\right).
\end{align*}
Finally, the result is a consequence of Lemma~\ref{lem:limimp} and Theorem~\ref{el13}.


\vspace{0.5cm}

\noindent\textbf{Aknowledgements.} 
Elisa  Al\`{o}s  acknowledges support from the Spanish Ministry of Economy and Competitiveness through the Severo Ochoa Programme for Centers of Excellence in R\&D (SEV-2015-0563), and grants ECO2014-59885-P and MTM2016-76420-P (MINECO/FEDER,
UE).
Jorge  A. Le\'on was partially supported by CONACyT grant 220303 and  thanks Universitat Barcelona, Universitat Pompeu Fabra and the Centre de Recerca Matem\`atica of  Universitat Aut\`onoma de Barcelona for their hospitality and financial support.
Antoine Jacquier acknowledges financial support from the EPSRC First Grant EP/M008436/1.

\end{document}